\documentclass{article}

\usepackage{amsmath,amsfonts,amssymb,amsthm}
\usepackage{fullpage}
\usepackage{url}

\newcommand{\Exp}[1]{\mathbb{E}\!\left[#1\right]}
   
\newcommand{\tr}{\mathrm{tr}}

\theoremstyle{definition}

\newtheorem{theorem}{Theorem}
\newtheorem{lemma}{Lemma}

\newcommand{\e}{\mathrm{e}}
\newcommand{\EC}{\textsc{Exact Cover}}
\newcommand{\POS}{\textsc{Positive 1-in-$k$ SAT}}
\newcommand{\dstar}{d^\star_k}
\newcommand{\wmax}{w_{\max}}

\begin{document}

\title{The phase transition in random regular exact cover}
\author{Cristopher Moore \\ Santa Fe Institute \\ \texttt{moore@santafe.edu}}
\maketitle

\abstract{
A $k$-uniform, $d$-regular instance of \textsc{Exact Cover} is a family of $m$ sets $F_{n,d,k} = \{ S_j \subseteq \{1,\ldots,n\} \}$, where each subset has size $k$ and each $1 \le i \le n$ is contained in $d$ of the $S_j$.  It is satisfiable if there is a subset $T \subseteq \{1,\ldots,n\}$ such that $|T \cap S_j|=1$ for all $j$.  Alternately, we can consider it a $d$-regular instance of \textsc{Positive 1-in-$k$ SAT}, i.e., a Boolean formula with $m$ clauses and $n$ variables where each clause contains $k$ variables and demands that exactly one of them is true.  We determine the satisfiability threshold for random instances of this type with $k > 2$.  Letting 
\[
d^\star 
= \frac{\ln k}{(k-1)(- \ln (1-1/k))} + 1
\, ,
\] 
we show that $F_{n,d,k}$ is satisfiable with high probability if $d < d^\star$ and unsatisfiable with high probability if $d > d^\star$.  We do this with a simple application of the first and second moment methods, boosting the probability of satisfiability below $d^\star$ to $1-o(1)$ using the small subgraph conditioning method.
}

\section{Introduction}

A $k$-uniform $d$-regular instance of \EC, or equivalently a \POS\ formula, has $n$ variables and $m$ clauses where $dn=km$.  We can treat it as a bipartite multigraph, with $n$ variables of degree $d$ on one side connected to $m$ clauses of degree $k$ on the other.  A satisfying assignmnent is a subset $T$ of the variables such that exactly one variable in each clause is true.

We choose random formulas $F_{n,d,k}$ according to the configuration model: that is, we make $d$ copies of each variable and $k$ copies of each clause, and choose a uniformly random bipartite matching of the resulting $dn=km$ copies with each other.  We assume that $d,k = O(1)$ so that $m=\Theta(n)$.

We determine the satisfiability threshold for these formulas.  Namely, we prove the following.
\begin{theorem}
\label{thm:main}
Let
\begin{equation}
\label{eq:dstar}
\dstar 
= \frac{\ln k}{(k-1)(- \ln (1-1/k))} + 1
\, .
\end{equation}
Then for any $k > 3$ and any integer $d$, 
\[
\lim_{n \to \infty} \Pr[\mbox{$F_{n,d,k}$ is satisfiable}] 
= \begin{cases} 
0 & d > \dstar \\ 
1 & d < \dstar \, .
\end{cases}
\]
\end{theorem}
\noindent
Note that when $k$ is large, $\dstar \approx (\ln k) + 1$.  Note also that $\dstar$ is never an integer, since then $k$ would be a rational power of $k-1$.

An easy application of the first and second moment method gives unsatisfability w.h.p. for $d > \dstar$, and satisfiability with positive probability for $d < \dstar$.  We boost the latter to high probability with the small subgraph conditioning method~\cite{robinson-wormald,wormald}.

The fact that the second moment method is exact suggests that, at least in the $d$-regular case, this problem does not have a condensation transition.  In contrast, for \textsc{Graph Coloring}, NAE-$k$-SAT and $k$-SAT, at a certain density condensation occurs~\cite{krzakala-pnas,coja-lenka,coja-coloring}: the set of satisfying assignments becomes dominated by a single cluster, and the number of satisfying assignments becomes much less concentrated.  Thus while the second moment method gives fairly good bounds for these problems~\cite{ach-moore,ach-peres,ach-naor,ach-moore-reg}, pushing it beyond this point requires much more sophisticated methods that count clusters of solutions, and further reduce the variance by carefully conditioning on the distribution of neighborhood structures throughout the formula~\cite{coja-nae,coja-ksat,ding-nae}.  This line of work recently culminated in a proof of the threshold conjecture for $k$-SAT for sufficiently large $k$~\cite{ding-ksat}, although many open questions still remain.  

Here the situation is much simpler.  The only source of variance in the number of satisfying assignments is the number of cycles of each length in the formula, so the small subgraph conditioning method reduces the variance enough to prove satisfiability with high probability.  It also turns out that that the point corresponding to two independent satisfying assignments is a local maximum of the rate function for the second moment, so there is no need to reweight the assignments as in~\cite{ach-peres,ach-naor-peres,dani-moore}.  

The second moment method also owes its success to the fact that, in the $d$-regular case, \POS\ is ``locked''  in the sense that most variables cannot be flipped without also flipping many others, so that satisfying assignments are isolated~\cite{locked1,locked2}.  Given a set $S \in \{1,\ldots,k-1\}$ let $S$-SAT be the problem where each clause has $k$ variables, and demands that the number of true variables it contains is an element of $S$.  If $S$ does not contain any adjacent pairs $i, i+1$, and if every variable has degree at least $2$, these problems are locked.  In~\cite{locked2} the authors wrote the first and second moments for this family of problems, described the resulting bound as a fixed point equation, and conjectured that it is exact.  This paper proves the case of their conjecture where $S=\{1\}$.

One can also consider random \POS\ formulas where clauses appear independently, so that the degrees of the variables are Poisson distributed.  A lower bound on the threshold in this model was given in~\cite{vamsi} for $k=3$ using differential equations.  Other constraint satisfaction problems for which the threshold can be computed exactly (and where condensation does not appear to occur) include random XOR-SAT~\cite{dubois-mandler,mrtz,pittel-sorkin} as well as 1-in-$k$ SAT~\cite{acim} where literals are negated with probability $1/2$ as opposed to the positive case we consider here.

We write $f(n) \sim g(n)$ if $\lim_{n \to \infty} f(n)/g(n) = 1$.  We say a series of events $E_n$ holds with high probability if $\Pr[E_n] \sim 1$, and with positive probability if, for some constant $B > 0$, $\Pr[E_n] \ge B$ for all sufficiently large $n$.

\section{The first and second moments}
\label{sec:moments}

In this section we compute the first and second moments, and show that $\Exp{Z^2} / \Exp{Z}^2$ tends to a constant.  This is enough to show that $F_{n,d,k}$ is satisfiable with positive probability for $d < \star$, and we improve this to high probability in the following section.  

\begin{lemma}
\label{lem:upper}
If $d > \dstar$ then $F_{n,d,k}$ is unsatisfiable with high probability.
\end{lemma}

\begin{proof}
Let $Z$ be the number of satisfying assignments $T$.  Since $|T| = n/k$, the expectation of $Z$ is ${n \choose n/k}$ times the fraction of matchings, for a given $T$, that connect each clause to exactly one of the $dn/k=m$ copies of variables in $T$.  Applying Stirling's formula $x! = (1+o(1)) \sqrt{2 \pi x} \,x^x \e^{-x}$ gives
\begin{equation}
\label{eq:first}
\Exp{Z} 
= {n \choose n/k} \frac{m! \,k^m \,((k-1)m)!}{(km)!} 
= k^m {n \choose n/k} \left/ {km \choose m} \right.
\sim \sqrt{d} \,k^m \,\e^{(n-m) h(1/k)}  
= \sqrt{d} \,\e^{n \phi_1} \, . 
\end{equation}
where 
\begin{align}
\phi_1 
&= \frac{d}{k} \ln k - (d-1) \,h(1/k) \nonumber \\
&= \frac{\ln k + (d-1)(k-1) \ln (1-1/k)}{k} 
\, . 
\label{eq:phi1}
\end{align}
and 
\[
h(\alpha) = -\alpha \ln \alpha - (1-\alpha) \ln (1-\alpha)
\]
denotes the Shannon entropy function.   
Since $\phi_1=0$ when $d = \dstar$, and $\phi_1$ is a decreasing function of $d$, we have $\Pr[Z > 0] \le \Exp{Z} = \e^{-\Omega(n)}$ whenever $d > \dstar$.
\end{proof}

\newpage

\begin{lemma}
If $k > 2$ and $d < \dstar$ then $F_{n,d,k}$ is satisfiable with positive probability.
\end{lemma}

\begin{proof}
The second moment is the expected number of pairs of assignments $T, T'$ that are both satisfying.  This depends on the size of their difference.  For a given $w \in [0,1]$, let $Z^{(2)}_w$ denote the expected number of satisfying pairs with $|T-T'| = |T'-T| = wn/k$.  For a given such pair, $(1-w)m$ of the clauses must be satisfied by a variable in $T \cap T'$, and the remaining $wm$ clauses must be satisfied both by a variable in $T-T'$ and one in $T'-T$.  The number of such matchings is
\begin{gather*}
{m \choose wm} ((1-w)m)! \,k^{(1-w)m} \,(wm)!^2 \,(k(k-1))^{wm} \,((k-(1-w)-2w)m)! \\
= k^m \,(k-1)^{wm} \,m! \,(wm)! \,((k-1-w)m)! \, .
\end{gather*}
Thus
\begin{align}
\Exp{Z^{(2)}_w}
&= k^m \,(k-1)^{wm} {n \choose n/k} {n/k \choose wn/k} {(1-1/k)n \choose wn/k} \frac{m! \,(wm)! \,((k-1-w)m)!}{(km)!} 
\nonumber \\
&= \Exp{Z} \,(k-1)^{wm} {n/k \choose wn/k} {(1-1/k)n \choose wn/k} \left/ {(k-1)m \choose wm} \right. 
\, .
\label{eq:expz2}
\end{align}
For $0 < w < 1$, applying Stirling's formula to~\eqref{eq:expz2} gives
\[
\Exp{Z^{(2)}_w}
\sim \frac{1}{\sqrt{2\pi n}} \,f(w) \,\e^{n \phi_2(w)} \, , 
\]
where 
\begin{equation}
\label{eq:f}
f(w) = d\sqrt{\frac{k}{w(1-w)}} 
\end{equation}
and
\begin{equation}
\label{eq:phi2}
\phi_2(w) 
= \phi_1 
+ \frac{wd}{k} \ln (k-1) 
+ \frac{1}{k} \,h(w) 
- (d-1) \left( 1-\frac{1}{k} \right) \,h\!\left( \frac{w}{k-1} \right) \, . 
\end{equation}

As in~\cite{ach-moore}, we can approximate the second moment by an integral, which we evaluate asymptotically using Laplace's method.  If $\phi_2(w)$ has a unique maximum $\wmax \in [0,1]$ where $0 < \wmax < 1$ and $\phi''_2(\wmax) < 0$, then 
\begin{align}
\Exp{Z^2} 
&= \sum_{w=0,k/n,2k/n,\ldots} \Exp{Z^{(2)}_w} \nonumber \\
&\sim \frac{1}{\sqrt{2\pi n}} \frac{n}{k} \int_0^1 \mathrm{d}w \,f(w) \,\e^{n \phi_2(w)} \nonumber \\
&\sim \frac{1}{k} \frac{f(\wmax)}{\sqrt{-\phi''_2(\wmax)}} \,\e^{n \phi_2(\wmax)} \, .
\label{eq:laplace}
\end{align}
In particular, suppose $\wmax = 1-1/k$.  We have
\[
\phi_2(1-1/k) = 2 \phi_1 \, ,
\]
which corresponds to the fact that $1-1/k$ is the typical value of $w$ if the two sets $T, T'$ are chosen independently.  Thus if $\phi_2$ is maximized at $1-1/k$, and if $\phi''_2 < 0$ there, we have $\Exp{Z^2} \sim C \Exp{Z}^2$ for some constant $C$.

The following lemma shows that this is in fact the case whenever $d < \dstar$.
\begin{lemma}  
\label{lem:wmax}
Let $k > 2$ and $d < \dstar$.  Then $\wmax = 1-1/k$ is the unique maximum of $\phi_2(w)$ in the unit interval, and $\phi''_2(\wmax) < 0$.
\end{lemma}

\begin{proof}
By direct calculation we have $\phi'_2(1-1/k) = 0$ and 
\[
\phi''_2(1-1/k) = - \frac{k(k-d)}{(k-1)^2} \, ,
\]
which is negative since $\dstar < k$ for all $k > 2$.  Thus $1-1/k$ is a local maximum.  To show that it is unique, note that $\phi_2$ has a unique inflection point $w_0$ where $\phi''_2 = 0$, namely 
\[
w_0 = \frac{(d-2)(k-1)}{dk-d-k} \, . 
\]
This implies that $1-1/k$ is the only local maximum.  Thus we just have to eliminate the possibility that the maximum of $\phi_2$ in the unit interval is at $w=0$ or $w=1$.  But this is easy: since $d < \dstar$ we have $\phi_1 > 0$, so $\phi_2(0) = \phi_1 < 2 \phi_1 = \phi_2(1-1/k)$, and as $w \to 1$.  At the other end of the interval, as $w \to 1$ we have $\phi_2(w) \to -\infty$ due to the $h(w)$ term in~\eqref{eq:phi2}.
\end{proof}

Plugging Lemma~\ref{lem:wmax} into the Laplace method~\eqref{eq:laplace} gives
\[
\Exp{Z^2} \sim d \sqrt{\frac{k-1}{k-d}} \,\e^{2n\phi_1} \, ,
\]
and combining this with~\eqref{eq:first} gives
\begin{equation}
\label{eq:ratio}
\frac{\Exp{Z^2}}{\Exp{Z}^2} 
\sim \sqrt{\frac{k-1}{k-d}} = C \, . 
\end{equation}
It follows that $\Pr[Z > 0] \ge 1/C$, completing the proof.
\end{proof}

\section{Small subgraph conditioning}

When there are strong correlations between the events that a pair of assignments are both satisfying, the variance $\Exp{Z^2}-\Exp{Z}^2$ is a constant times $\Exp{Z}^2$, and the second moment method can only prove satisfiability with positive probability.  However, in some cases we can show that the variance is much smaller if we condition on the number of small subgraphs in the formula---in particular, the number of cycles of each constant length.  This technique was introduced in~\cite{robinson-wormald}, where it was used to show that random 3-regular graphs possess a Hamiltonian cycle with high probability; another application~\cite{diaz-5col} showed that random 5-regular graphs are 3-colorable with high probability.

Let $X_i$ be the number of cycles of length $2i$ in the formula, i.e., cycles alternating between $i$ distinct variables and $i$ distinct clauses.  Our goal is to compute the correlation between $Z$ and $X_i$ and its higher moments, and hence to learn to what extent $X_i$ affects the number of satisfying assignments.  Our goal is to explain almost all of the variance in $Z$ with the variance in the $X_i$.

Let $(x)_r$ denote the falling factorial $x(x-1)(x-2) \cdots (x-r+1)$; thus $(X_i)_r$ is the number of ordered lists of $r$ cycles of length $2i$.  
If $X$ is Poisson with mean $\lambda$, we have $\Exp{(X)_r} = \lambda^r$.  We use the following ``plug and play'' version of the subgraph conditioning method from~\cite{wormald}.

\begin{theorem}
Let $Z$ and $X_1, X_2, \ldots$ be nonnegative integer-valued random variables.  Suppose that $\Exp{Z} > 0$, and that for each $i \ge 0$ there are constants $\lambda_i > 0$, $\delta_i > -1$ such that
\begin{enumerate}

\item For any $j$, the variables $X_1, \ldots, X_j$ are asymptotically independent and Poisson distributed, with $\Exp{X_i} \sim \lambda_i$,
\item For any sequence $m_1, \ldots, m_j$ of nonnegative integers, 
\begin{equation}
\label{eq:subgraph-moments}
\frac{\Exp{Z \prod_{i=1}^j (X_i)_{m_i}}}{\Exp{Z}} 
\sim \prod_{i=1}^j \mu_i^{m_i} 
\quad \text{where} \quad 
\mu_i = \lambda_i (1+\delta_i) \, ,
\end{equation}
\item $\sum_{i=1}^\infty \lambda_i \delta_i^2$ is finite, and
\begin{equation}
\label{eq:subgraph-second}
\frac{\Exp{Z^2}}{\Exp{Z}^2} 
\sim \exp\!\left( \sum_{i=1}^\infty \lambda_i \delta_i^2 \right) \, . 
\end{equation}
\end{enumerate}
Then $\Pr[Z > 0] = 1-o(1)$.
\end{theorem} 

Applying this technology to prove the following theorem, and thus complete the proof of Theorem~\ref{thm:main}, is an enjoyable exercise in combinatorics.

\begin{theorem}
\label{thm:high}
If $k > 2$ and $d < \dstar$ then $F_{n,d,k}$ is satisfiable with high probability. 
\end{theorem}

\begin{proof}
Standard arguments for sparse random graphs~\cite{bollobas} show that the $X_i$ are asymptotically independent and Poisson distributed.  To compute the asymptotic expectation $\lambda_i$, note that there are $(m)_i (n)i$ sequences of clauses and variables that $C$ could visit; since there are $i$ variables where we could start a cycle and two directions in which we could go, this overcounts by a factor of $2i$.  There are $(k(k-1)d(d-1))^i$ choices of copies with which to wire each variable to the clause before and after it in the sequence, and the number of matchings that include a given such wiring is $(km-2i)!$.  Thus
\begin{equation}
\label{eq:lambda}
\Exp{X_i} 
= \frac{1}{2i} (m)_i (n)_i \big( k(k-1)d(d-1) \big)^i \,\frac{(km-2i)!}{(km)!} 
\sim \frac{\big( (k-1) (d-1) \big)^i}{2i} 
= \lambda_i \, .
\end{equation}

In order to establish~\eqref{eq:subgraph-moments}, we first warm up by computing $\Exp{Z X_i}$.  This is the sum over all pairs $(T, C)$, where $T$ is an assignment and $C$ is a cycle of length $2i$, of the fraction of matchings containing $C$ for which $T$ is satisfying.  

We start by choosing one of the ${n \choose n/k}$ possible satisfying assignmments $T$.  We then choose $C$.  First, we choose $t=|C \cap T|$, the number of true variables in $C$.  Let us think of $C$ as a cycle of $i$ variables, where the edges between them correspond to their shared clauses.  Since each clause must contain exactly one true variable, none of $C$'s true variables can be adjacent; in particular, $t \le \lfloor i/2 \rfloor$.  (This is similar to~\cite{robinson-wormald}, where no two adjacent edges of $C$ can belong to a Hamiltonian cycle.)  Let $N_{i,t}$ be the number of ordered, labeled cycles with $t$ true variables, where no two true variables are adjacent; for instance, $N_{6,0}=1$, $N_{6,1} = 6$, $N_{6,2} = 9$, and $N_{6,3} = 2$.

Now that we have chosen $t$, and chosen one of the $N_{i,t}$ arrangements of true variables in it, we choose what variables and clauses $C$ contains and how they are matched to each other.  There are $(m)_i$ ordered sets of $i$ clauses, and $(n/k)_t \,\big( (1-1/k)n \big)_{i-t}$ choices of which true and false variables appear in $C$ and in what order.  As before, there are $( k(k-1) d(d-1) )^i$ ways to wire each variable to the clause before and after it, and all this overcounts by a factor of $2i$.

At this point in the process, we have already satisfied $2t$ clauses in $C$, so there are $m-2t$ clauses waiting to be satisfied.  Happily, we have $dn/k-2t = m-2t$ unmatched copies of true variables with which to satisfy them.  The $m-i$ clauses outside $C$ have $k$ unmatched copies each, and the $i-2t$ clauses in $C$ that are not yet satisfied each have $k-2$ unmatched copies.  Thus there are $(m-2t)!$ orders in which we can assign copies of true variables to clauses, and $k^{m-i} \,(k-2)^{i-2t}$ ways to match them with these clauses' copies.  After all this, there are $(k-1)m-2(i-t)$ unmatched copies of false variables, which can be matched with the remaining clause copies arbitrarily.  Finally, we divide by $(km)!$ to obtain
\begin{align*}
\Exp{Z X_i}
= {n \choose n/k} 
\sum_{t=0}^{\lfloor i/2 \rfloor} &\left[ \frac{N_{i,t}}{2i} \,(m)_i \,(n/k)_t \,\big( (1-1/k)n \big)_{i-t} 
\big( k(k-1) d(d-1) \big)^i \right. \\
& \times \left. \frac{(m-2t)! \,k^{m-i} (k-2)^{i-2t} \big( (k-1)m - 2(i-t)\big)!}{(km)!} \right] \, .
\end{align*}
Dividing by $\Exp{Z}$ and using $(m)_i \sim m^i$, $m! / (m-2t)! \sim m^{2t}$ and so on gives
\begin{align*}
\frac{\Exp{Z X_i}}{\Exp{Z}}
&= \sum_{t=0}^{\lfloor i/2 \rfloor} \left[ \frac{N_{i,t}}{2i} 
\,(m)_i \,(n/k)_t \,\big( (1-1/k)n \big)_{i-t} 
\big( k(k-1) d(d-1) \big)^i \right. \\
& \qquad \quad \times \left. \frac{(m-2t)! \,k^{m-i} (k-2)^{i-2t} \big( (k-1)m - 2(i-t)\big)!}{m! \,k^m \,((k-1)m)!} \right] \\
&\sim 
\frac{\big( (k-2)(d-1) \big)^i}{2i}
\sum_{t=0}^{\lfloor i/2 \rfloor} N_{i,t} \left( \frac{k-1}{(k-2)^2} \right)^{\!t} \\
&= \mu_i = \lambda_i (1+\delta_i) \, , 
\end{align*}
where
\[
\delta_i = \left( \frac{k-2}{k-1} \right)^{\!i}
\,\sum_{t=0}^{\lfloor i/2 \rfloor} N_{i,t} \left( \frac{k-1}{(k-2)^2} \right)^{\!t} 
\;-\; 1
\, . 
\]
We can evaluate this sum with the generating function
\begin{align*}
g(z) = \sum_{t=0}^{\lfloor i/2 \rfloor} N_{i,t} z^t 
&= \tr \begin{pmatrix} 0 & \sqrt{z} \\ \sqrt{z} & 1 \end{pmatrix}^{\!i} \\
&= \left( \frac{1+\sqrt{1+4z}}{2} \right)^{\!i}
+ \left( \frac{1-\sqrt{1+4z}}{2} \right)^{\!i} \, ,
\end{align*}
giving
\begin{equation}
\label{eq:delta}
\delta_i 
= \left( \frac{k-2}{k-1} \right)^{\!i} g\!\left( \frac{k-1}{(k-2)^2} \right) - 1
= \left( -\frac{1}{k-1} \right)^i \, .
\end{equation}

Generalizing this calculation to show that~\eqref{eq:subgraph-moments} holds is a matter of bookkeeping.  Let $\ell = \sum_{s=1}^j m_i$, and let $i_1,\ldots,i_\ell$ be a sorted list where each $s$ appears $m_s$ times.  Then $\Exp{Z \prod_{i=1}^j (X_i)_{m_i}}$ is the expected number of tuples $(T, C_1, \ldots, C_\ell)$ where $T$ is a satisfying assignment and each $C_s$ is a cycle of length $2 i_s$.   Counting as before gives
\begin{align*}
\frac{\Exp{Z \prod_{i=1}^j (X_i)_{m_i}}}{\Exp{Z}} 
&= \sum_{t_1=0}^{\lfloor i_1/2 \rfloor} \sum_{t_2=0}^{\lfloor i_2/2 \rfloor} \cdots \sum_{t_\ell=0}^{\lfloor i_\ell/2 \rfloor} 
\left[ \prod_{s=1}^\ell \frac{N_{i_s,t_s}}{2i_s} \right. \\
& \quad \times \left. \,(m)_{\sum_s \!i_s} \,(n/k)_{\sum_s \!t_s} \,\big( (1-1/k)n \big)_{\sum_s (i_s-t_s)}
\big( k(k-1) d(d-1) \big)^{\sum_s \!i_s} \right] \nonumber \\
& \quad \times \left. \frac{(m-2\sum_s t_s)! \,k^{m- \sum_s \!i_s} (k-2)^{\sum_s (i_s-2t_s)} \big( (k-1)m - 2 \sum_s (i_s-t_s)\big)!}{m! \,k^m \,((k-1)m)!} \right] \\
&\sim 
\prod_{s=1}^\ell 
\frac{\big( (k-2)(d-1) \big)^{i_s}}{2i_s}
\sum_{t_s=0}^{\lfloor i_s/2 \rfloor} N_{i_s,t_s} \left( \frac{k-1}{(k-2)^2} \right)^{\!t_s} \\
&= \prod_{s=1}^\ell \mu_{i_s} 
= \prod_{i=1}^j \mu_i^{m_i} \, . 
\end{align*}

Finally, we establish~\eqref{eq:subgraph-second}.  Using the Taylor series $-\log (1-z) = \sum_{i=1}^\infty z^i / i$ gives
\[
\sum_{i=1}^\infty \lambda_i \delta_i^2 
= \frac{1}{2} \sum_{i=1}^\infty \frac{1}{i} \left( \frac{d-1}{k-1} \right)^{\!i}
= \frac{1}{2} \log \frac{k-1}{k-d} \, , 
\]
and comparing with~\eqref{eq:ratio} shows that this is indeed the logarithm of the asymptotic ratio $C \sim \Exp{Z^2} / \Exp{Z}^2$.  This completes the proof.
\end{proof}

\section*{Acknowledgments}

This work was supported by NSF grants CCF-1117426 and CCF-1219117.  I am grateful to Allan Sly, Lenka Zdeborov\'a, and Amin Coja-Oghlan for helpful conversations.  I am also grateful to the Bellairs Research Institute of McGill University where part of this work was carried out.


\begin{thebibliography}{99}

\bibitem{robinson-wormald} R.W. Robinson and N.C. Wormald, 
``Almost all cubic graphs are Hamiltonian.'' 
\emph{Random Structures and Algorithms} 3: 117--125 (1992).

\bibitem{wormald} Nicholas C. Wormald, ``Models of random regular graphs.'' 
J.D. Lamb and D.A. Preece, Eds., \emph{Surveys in Combinatorics,}  
LMS Lecture Note Series 267. 
Cambridge University Press, 1999, 239--298.

\bibitem{krzakala-pnas} F. Krzakala, A. Montanari, F. Ricci-Tersenghi, G. Semerjian, and L. Zdeborov\'a, 
``Gibbs states and the set of solutions of random constraint satisfaction problems.'' \emph{Proc. Natl. Acad. Sci.} 
104(25):10318--10323 (2007).

\bibitem{coja-lenka} Amin Coja-Oghlan and Lenka Zdeborov\'a, 
``The condensation transition in random hypergraph 2-coloring.''  SODA 2012: 241--250.

\bibitem{coja-coloring} Victor Bapst, Amin Coja-Oghlan, Samuel Hetterich, Felicia Ra{\ss}mann, and Dan Vilenchik, 
``The Condensation Phase Transition in Random Graph Coloring.''  APPROX-RANDOM 2014: 449--464.

\bibitem{ach-moore} Dimitris Achlioptas and Cristopher Moore, ``Two moments suffice to cross a sharp threshold.''  \emph{SIAM J. Computing} 36:740--762 (2006).

\bibitem{ach-peres} Dimitris Achlioptas and Yuval Peres, 
``The threshold for random $k$-SAT is $2^k (\ln 2 - O(k))$.'' STOC 2003: 223--231.

\bibitem{ach-naor} Dimitris Achlioptas and Assaf Naor, 
``The two possible values of the chromatic number of a random graph.'' STOC 2004: 587--593.

\bibitem{ach-moore-reg} Dimitris Achlioptas and Cristopher Moore, 
``The Chromatic Number of Random Regular Graphs.'' APPROX-RANDOM 2004: 219--228

\bibitem{coja-nae} Amin Coja-Oghlan and Konstantinos Panagiotou, 
``Catching the $k$-NAESAT threshold.'' STOC 2012: 899--908.

\bibitem{coja-ksat} Amin Coja-Oghlan,
``The asymptotic $k$-SAT threshold.'' STOC 2014: 804--813.

\bibitem{ding-nae} Jian Ding, Allan Sly, and Nike Sun, 
``Satisfiability threshold for random regular NAE-SAT.'' STOC 2014: 814--822.

\bibitem{ding-ksat} Jian Ding, Allan Sly, and Nike Sun, 
``Proof of the satisfiability conjecture for large $k$.''  Preprint, \url{http://arxiv.org/abs/1411.0650}. 

\bibitem{locked1} Lenka Zdeborov\'a and Marc M\'ezard, 
``Locked constraint satisfaction problems.'' 
\emph{Phys. Rev. Lett.} 101: 078702 (2008).

\bibitem{locked2} Lenka Zdeborov\'a and Marc M\'ezard, 
``Constraint satisfaction problems with isolated solutions are hard.''
\emph{J. Stat. Mech.} P12004 (2008).

\bibitem{ach-naor-peres} Dimitris Achlioptas, Assaf Naor, and Yuval Peres, 
``On the maximum satisfiability of random formulas.'' FOCS 2003: 362--370.

\bibitem{dani-moore} Varsha Dani and Cristopher Moore,
``Independent Sets in Random Graphs from the Weighted Second Moment Method.'' APPROX-RANDOM 2011: 472--482.

\bibitem{vamsi} Vamsi Kalapala and Cristopher Moore, 
``The Phase Transition in Exact Cover.''   
\emph{Chicago Journal of Theoretical Computer Science} 5 (2008).

\bibitem{dubois-mandler} O. Dubois and J. Mandler, 
``The $3$-XORSAT threshold.'' FOCS 2002: 769--778.

\bibitem{mrtz} M. M\'ezard, F. Ricci-Tersenghi, and R. Zecchina. 
``Two solutions to diluted $p$-spin models and XORSAT problems.''
\emph{J. Stat. Phys.} 111:505--533 (2003).

\bibitem{pittel-sorkin} B. Pittel and G. B. Sorkin, 
``The satisfiability threshold for $k$-XORSAT.''  Preprint, \url{http://arxiv.org/abs/1212.3822v2} (2012).

\bibitem{acim} D. Achlioptas, A. Chtcherba, G. Istrate, and C. Moore, 
``The phase transition in 1-in-$k$ SAT and NAE 3-SAT.''  
SODA 2001: 721--722.

\bibitem{diaz-5col} Josep D\'iaz, Alexis C. Kaporis, G. D. Kemkes, Lefteris M. Kirousis, Xavier P\'erez, and Nicholas C. Wormald, 
``On the chromatic number of a random 5-regular graph.'' \emph{Journal of Graph Theory} 61(3): 157-191 (2009).

\bibitem{bollobas} B. Bollob\'as, \emph{Random Graphs}.  Academic Press, 1985.





\end{thebibliography}
\end{document}